\documentclass[a4paper]{article}

\usepackage{amsmath,amssymb}
\usepackage{amsthm}
\usepackage{tabls}
\usepackage{graphicx}
\usepackage{wrapfig}
\usepackage{array}
\usepackage{url}
\usepackage{multirow}
\usepackage{multicol}
\usepackage{here}
\usepackage[margin=1in]{geometry}
\usepackage[algoruled,linesnumbered,algo2e,vlined]{algorithm2e}
\SetArgSty{textrm}

\newtheorem{theorem}{Theorem}
\newtheorem{definition}[theorem]{Definition}
\newtheorem{lemma}[theorem]{Lemma}

\newtheorem{observation}[theorem]{Observation}

\newtheorem{problem}[theorem]{Problem}

\title{
  Faster STR-EC-LCS Computation
}

\author{
  Kohei~Yamada$^1$\quad
  Yuto~Nakashima$^1$\quad
  Shunsuke~Inenaga$^{1, 2}$\quad\\
  Hideo~Bannai$^1$\quad
  Masayuki~Takeda$^1$\\
  {$^1$ Department of Informatics, Kyushu University, Fukuoka, Japan}\\
  {$^2$ PRESTO, Japan Science and Technology Agency, Kawaguchi, Japan}\\
  {\texttt{\{kohei.yamada,yuto.nakashima,inenaga}}\\
  {\texttt{bannai,takeda\}@inf.kyushu-u.ac.jp}}
}

\newcommand{\streclcs}{STR-EC-LCS}
\newcommand{\overlap}[1]{\mathit{\sigma(#1)}}
\newcommand{\prop}{\mathsf{Property}}
\newcommand{\nextchar}{\mathsf{next}_{\mathit{B}}}
\newcommand{\newoverlap}{\mathsf{next}_{\mathit{\sigma}}}
\newcommand{\kmp}{\mathit{kmp}}

\begin{document}
\maketitle

\begin{abstract}
  The longest common subsequence (LCS) problem is a central problem in stringology 
  that finds the longest common subsequence of given two strings $A$ and $B$.
  More recently, a set of four constrained LCS problems 
  (called generalized constrained LCS problem) were proposed by Chen and Chao [J. Comb. Optim, 2011].
  In this paper, we consider the substring-excluding constrained LCS (STR-EC-LCS) problem.
  A string $Z$ is said to be {\em an STR-EC-LCS of two given strings} $A$ and $B$ {\em excluding} $P$ if, 
  $Z$ is one of the longest common subsequences of $A$ and $B$ that does not contain $P$ as a substring.
  Wang et al. proposed a dynamic programming solution 
  which computes an STR-EC-LCS in $O(mnr)$ time and space 
  where $m = |A|, n = |B|, r = |P|$ [Inf. Process. Lett., 2013].
  In this paper, we show a new solution for the STR-EC-LCS problem.
  Our algorithm computes an STR-EC-LCS in $O(n|\Sigma| + (L+1)(m-L+1)r)$ time 
  where $|\Sigma| \leq \min\{m, n\}$ denotes the set of distinct characters
  occurring in both $A$ and $B$,
  and $L$ is the length of the STR-EC-LCS.
  This algorithm is faster than the $O(mnr)$-time algorithm for short/long STR-EC-LCS
  (namely, $L \in O(1)$ or $m-L \in O(1)$),
  and is at least as efficient as the $O(mnr)$-time algorithm for all cases.
\end{abstract}

\section{Introduction}

The \emph{longest common subsequence} (\emph{LCS}) problem
of finding an LCS of given two strings,
is a classical and important problem in Theoretical Computer Science.
Given two strings $A$ and $B$ of respective lengths $m$ and $n$,
it is well known that the LCS of $A$ and $B$ can be computed
by a standard dynamic programming technique~\cite{Wagner_1974_LCS}.
Since LCS is one of the most fundamental similarity measures for string comparison,
there are a number of studies on faster computation of LCS and its
applications~\cite{BUNKE199593,LCS_RLE_2014,LCSapp_bioinf_2011,LCSapp_patrecog_2013}.
It is also known that there is a conditional lower bound
which states that the LCS of two strings of length $n$ each
cannot be computed in $O(n^{2-\epsilon})$ time for any constant $\epsilon > 0$,
unless the famous popular Strong Exponential Time Hypothesis (SETH) fails~\cite{AbboudBW15}.
Thus, it is highly likely that one needs to use almost quadratic time
for computing LCS in the worst case.
Still, it is possible to design algorithms for computing LCS
whose running time depends on other parameters.
One of such algorithms was proposed by Nakatsu et al.~\cite{DBLP:journals/acta/NakatsuKY82},
which finds an LCS of given two strings $A$ and $B$ in $O(n(m-l))$ time and space, 
where $l$ is the length of the LCS of the two given strings.
This algorithm is efficient when $l$ is large, namely, $A$ and $B$ are very similar.

Of a variety of extensions to LCS that have been extensively studied,
this paper focuses on a class of problems called the \emph{constrained LCS} problems,
first considered by Tsai~\cite{CLCS_Tsai_2003}.
We are given strings $A, B$ and constraint string $P$ of length $r$, 
and the CLCS problem is to find a longest subsequence common to $A$ and $B$,
such that the subsequence has $P$ as a subsequence.
He also presented a dynamic programming algorithm which solves
the problem in $O(m^2n^2r)$ time and space.
The motivation for introducing constraints is to reflect some
a-priori knowledge (e.g., biological knowledge) to the solutions.
Later, the \emph{generalized constrained LCS} (\emph{GC-LCS})
problems were introduced by Chen et al.~\cite{DBLP:journals/jco/ChenC11}.
GC-LCS consists of four variants of the constrained LCS problem,
which are respectively called \emph{SEQ-IC-LCS}, \emph{SEQ-EC-LCS}, \emph{STR-IC-LCS}, and
\emph{STR-EC-LCS}.
For given strings $A, B$ and $P$, 
the problem is to find a longest subsequence common
to $A$ and $B$ such that 
the subsequence includes/excludes/includes/excludes $P$ as a subsequence/subsequence/substring/substring, respectively for SEQ-IC-LCS/SEQ-EC-LCS/STR-IC-LCS/STR-EC-LCS.
We remark that CLCS is the same as SEQ-IC-LCS.
The best known results for these problems were proposed in~\cite{SEQICLCS_2004,SEQECLCS_Chen_2011,STRICLCS_DEOROWICZ_2012,STRECLCS_Wang_2013}.

The quadratic bound for STR-IC-LCS seems to be very difficult to improve,
since STR-IC-LCS is a special case of LCS (recall the afore-mentioned conditional lower bound
for LCS).
Since the other three variants require cubic time,
it is important to discover more efficient solutions for these problems.
There exist faster dynamic programming solutions for SEQ-IC-LCS and STR-IC-LCS
which are based on run-length encodings~\cite{SEQ-IC-LCS-RLE_2014,STR-IC-LCS-RLE_Kuboi_2017}.
However, no faster solutions to STR-EC-LCS than the one
with $O(mnr)$ running time~\cite{STRECLCS_Wang_2013} are known to date.

In this paper, we revisit the STR-EC-LCS problem.
More formally, 
we say that a string $Z$ is {\em an STR-EC-LCS of two given strings} $A$ and $B$ {\em excluding} P if, 
$Z$ is one of the longest common subsequences of $A$ and $B$ that does not contain $P$ as a substring.
We show a new dynamic programming solution for the STR-EC-LCS problem 
which runs in $O(n|\Sigma| + (L+1)(m-L+1)r)$ time and space, 
where $\Sigma$ is the set of distinct characters occurring in both $A$ and $B$,
and $L$ is the length of the solution.
Note that $|\Sigma| \leq \min\{m, n\}$ always holds.
Our algorithm is built on Nakatsu et al.s' method for the (original) LCS problem~\cite{DBLP:journals/acta/NakatsuKY82}.
Assume w.l.o.g. that $m \leq n$.
When the length of STR-EC-LCS is quite short or long (namely, $L \in O(1)$ or $m-L \in O(1)$),
our algorithm runs only in $O(n|\Sigma|+mr) = O((n+r)m) = O(nm)$ time and space,
since $r \leq n$.
Even in the worst case where $L \in \Theta(m)$ and $m-L \in \Theta(m)$,
which happens when $L = cm$ for any constant $0 < c < 1$,
our algorithm is still as efficient as $O(mnr)$ since $|\Sigma| \leq \min\{m, n\}$.

This paper is organized as follows;
we will give notations which we use in this paper in Section~\ref{sec:preliminaries},
we will propose our dynamic programming solution for the STR-EC-LCS problem in Section~\ref{sec:dp-part},
finally, we will explain our algorithm for the STR-EC-LCS in Section~\ref{sec:algorithm}.

\section{Preliminaries}\label{sec:preliminaries}

\subsection{Strings}
Let $\Sigma$ be an integer {\em alphabet}.
An element of $\Sigma^*$ is called a {\em string}.
The length of a string $w$ is denoted by $|w|$.
The empty string $\varepsilon$ is a string of length 0.
For a string $w = xyz$, $x$, $y$ and $z$ are called
a \emph{prefix}, \emph{substring}, and \emph{suffix} of $w$, respectively.
The $i$-th character of a string $w$ is denoted by $w[i]$, where $1 \leq i \leq |w|$.
For a string $w$ and two integers $1 \leq i \leq j \leq |w|$,
let $w[i..j]$ denote the substring of $w$ that begins at position $i$ and ends at
position $j$. For convenience, let $w[i..j] = \varepsilon$ when $i > j$.

A string $Z$ is a {\em subsequence} of $A$ 
if $Z$ can be obtained from $A$ by removing zero or more characters.
In this paper, we consider common subsequences of two strings $A$ and $B$
of respective lengths $m$ and $n$.
For this sake, we can perform a standard preprocessing on $A$ and $B$
that removes every character that occurs only in either $A$ or $B$,
because such a character is never contained in any common subsequences of $A$ and $B$.
Assuming $n \geq m$,
this preprocessing can be done in $O(n \log n)$ time with $O(n)$ space
for general ordered alphabets,
and in $O(n)$ time and space for integer alphabets of polynomial size in $n$
(c.f. \cite{InenagaH18}).
In what follows, we consider the latter case of integer alphabets,
and assume that $A$ and $B$ have been preprocessed as above.
In the sequel, let $\Sigma$ denote the set of distinct characters
that occur in both $A$ and $B$.
Note that $|\Sigma| \leq \min\{m, n\} = m$ holds.

\subsection{STR-EC-LCS}
Let $A, B$ and $P$ be strings.
A string $Z$ is said to be {\em an STR-EC-LCS of two given strings} $A$ and $B$ {\em excluding} P if, 
$Z$ is one of the longest common subsequences of $A$ and $B$ that does not contain $P$ as a substring.
For instance, $\mathtt{bcaac}$, $\mathtt{bcaba}$, $\mathtt{acaac}$, $\mathtt{acaba}$, $\mathtt{abaac}$ and $\mathtt{ababa}$ 
are STR-EC-LCS of $A = \mathtt{abcabac}$ and $B = \mathtt{acbcaacbaa}$ excluding $P = \mathtt{abc}$.
Although $\mathtt{abcaba}$ and $\mathtt{abcaac}$ are longest common subsequences of $A$ and $B$,
they are not STR-EC-LCS of the same strings (since they have $P$ as a substring).

In Section~\ref{sec:dp-part}, we revisit the STR-EC-LCS problem defined as follows.

\begin{problem}[STR-EC-LCS problem~\cite{DBLP:journals/jco/ChenC11}]
	Given strings $A, B$, and $P$, 
	compute an STR-EC-LCS (and/or its length) of given strings.
\end{problem}

In the rest of the paper, $m, n$, and $r$
respectively denote the length of $A, B$ and $P$.
It is easy to see that STR-EC-LCS problem is the same as LCS problem when $r > \min\{m, n\}$.
We assume that $r \leq m \leq n$ without loss of generality.

\section{Dynamic programming solution for the STR-EC-LCS problem}\label{sec:dp-part}

Our aim of this section is to show our dynamic programming solution for the STR-EC-LCS problem.
We first give short descriptions of 
a dynamic programming solution for the LCS problem proposed by Nakatsu et al.~\cite{DBLP:journals/acta/NakatsuKY82}, 
and a dynamic programming solution for the STR-EC-LCS problem proposed by Wang et al.~\cite{STRECLCS_Wang_2013}.

\subsection{Solution for LCS by Nakatsu et al.}\label{subsec:Nakatsu}
Nakatsu et al. proposed a dynamic programming solution 
for computing an LCS of given strings $A$ and $B$.
Here, we give a slightly modified description
of their solution in order to describe our algorithm.
For any $0 \leq i,s \leq m$, let $e(i,s)$ be the length of the shortest prefix $B[1..e(i,s)]$ of $B$ 
such that the length of the longest common subsequence of $A[1..i]$ and $B[1..e(i,s)]$ is $s$.
For convenience, $e(i,s) = n+1$ if no such 
prefix exists or if $s > i$ holds.
The values $e(i,s)$ will be computed using dynamic programming, where $i$ represents the column number, and $s$ represents the row number.
Let $\tilde{s}$ be the largest value such that
$e(i,s) < n+1$ for some $i$, i.e,
$\tilde{s}$ is the last row in the table
of $e$, which has a value smaller than $n+1$.
We can see that the length of the longest common subsequence of $A$ and $B$ is $\tilde{s}$.
We give an example in Fig.~\ref{fig:LCS-dp}.

\begin{figure}[t]
    \centerline{\includegraphics[width=0.4\linewidth]{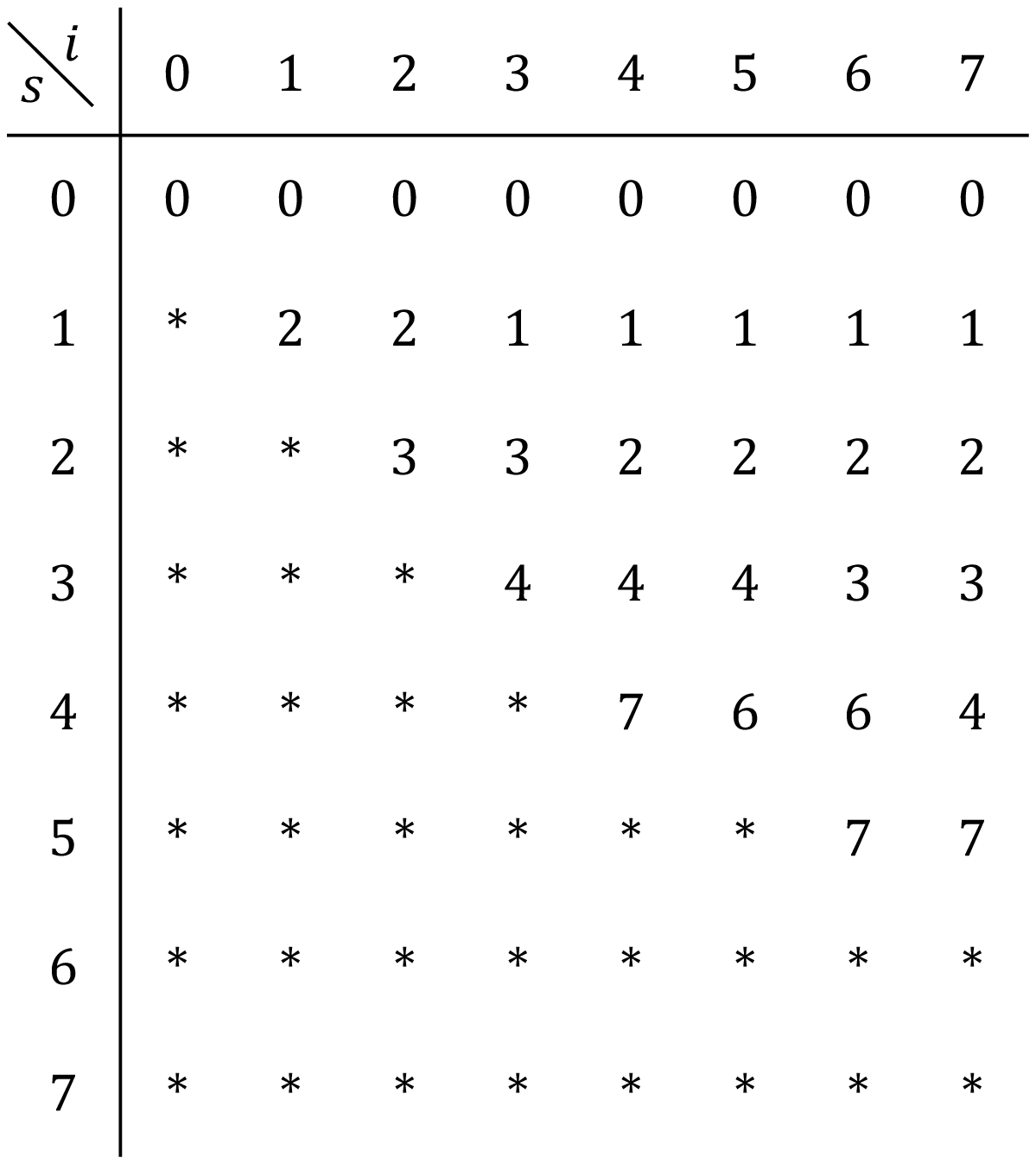}}
    \caption{
    	This is an example for table $e$ of given strings
    	$A = \mathtt{aabacab}$ and $B = \mathtt{baabbcaa}$.
    	For the sake of visibility, the value $n+1 = 9$ is replaced by asterisk ($*$).
    	The last row in the table which has a value smaller than $n+1$ is $5$;
    	that is, the length of an LCS of $A$ and $B$ is $5$.
    }\label{fig:LCS-dp}
\end{figure}

Now we explain how to compute $e$ efficiently.
Assume that $e(i-1,s)$ and $e(i-1,s-1)$ have already been computed.
We consider $e(i,s)$.
It is easy to see that $e(i,s) \leq e(i-1,s)$.
If $e(i,s) < e(i-1,s)$, an LCS of $A[1..i]$ and $B[1..e(i,s)]$ 
must use the character $A[i]$ as the last character.
Then, we can see that $e(i,s)$ is the index of the leftmost occurrence of $A[i]$ in $B[e(i-1,s-1)+1..n]$.
Let $j_{i,s}$ be the the index of the leftmost occurrence of $A[i]$ in $B[e(i-1,s-1)+1..n]$.
From these facts, the following recurrence formula holds
for $e$:
\begin{equation*}
	e(i,s) = \min\{e(i-1,s), j_{i,s}\}.
\end{equation*}

If we add more information, 
we can backtrack on the table in order to compute an LCS (as a string), and not just its length.

\subsection{Solution for STR-EC-LCS by Wang et al.}\label{subsec:Wang}
Wang et al. proposed a dynamic programming solution for STR-EC-LCS problem of given strings $A, B$ and $P$.
Here, we describe a key idea of their solution.

\begin{definition}
	For any string $S$, 
	$\overlap{S}$ is the length of the longest prefix of $P$
	which is a suffix of $S$.
\end{definition}

By using this notation, they considered a table $f$ defined as follows:
let $f(i,j,k)$ be the length of the longest common subsequence $Z$ of $A[1..i]$ and $B[1..j]$ 
such that $Z$ does not have $P$ as a substring and $\overlap{Z} = k$.
They also showed a recurrence formula for $f$.
By the definition of $f$, the length of an STR-EC-LCS is $\max\{f(m,n,t) \mid 0 \leq t < r\}$.

\subsection{Our solution for STR-EC-LCS}
Our solution is based on the idea of Section~\ref{subsec:Nakatsu}.
We maintain occurrences of a prefix of $P$ as a suffix of a common subsequence 
by using the idea of Section~\ref{subsec:Wang}.

For convenience, we introduce the following notation.
\begin{definition}
	A string $Z$ is said to satisfy $\prop(i,s,k)$
	if
	\begin{itemize}
		\item $Z$ is a subsequence of $A[1..i]$,
		\item $Z$ does not have $P$ as a substring,
		\item $|Z| = s$, and 
		\item $\overlap{Z} = k$.
	\end{itemize}
\end{definition}

Thanks to the above notation, 
we can simply introduce our table $d$ for computing STR-EC-LCS as follows.
Let $d$ be a 3-dimensional table where $d(i,s,k)$ is the length of the shortest prefix $B[1..d(i,s,k)]$ of $B$ 
such that there exists a subsequence which satisfies $\prop(i,s,k)$ 
(if no such subsequence exists, then $d(i,s,k) = n+1$ for convenience).

We can obtain the following observation about the length of an \streclcs~by the definition of $d$.

\begin{observation}
	Let $\tilde{s}$ be the largest $1 \leq s \leq m$ 
	such that $d(i,s,k) < n+1$ for some $i$ and $k$.
	$\tilde{s}$ is the length of an \streclcs~by the definition of $d$.
\end{observation}

We give an example of a table in Fig.~\ref{fig:dptables}.

\begin{figure}[t]
    \centerline{\includegraphics[width=0.8\linewidth]{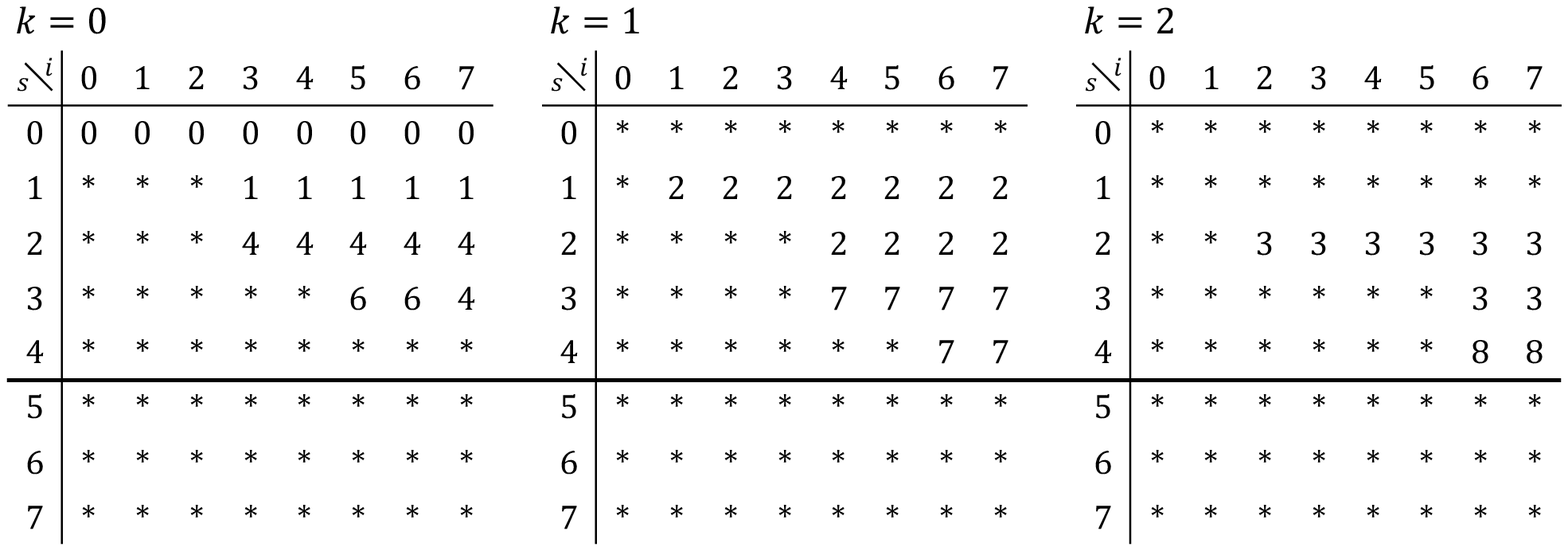}}
    \caption{
    	This is our table $d$ for given strings
    	$A = \mathtt{aabacab}, B = \mathtt{baabbcaa}$, and $P = \mathtt{aab}$.
    	In this figure, the value $n+1 = 9$ is replaced by asterisk ($*$) for convenience.
    	The lowest row which has a value smaller than $n+1 = 9$ is $\tilde{s} = 4$.
    	Thus, the length of a STR-EC-LCS is $4$.
    }\label{fig:dptables}
\end{figure}

The next lemma shows a recurrence formula for $d$.
We use this lemma for computing the length of a STR-EC-LCS.

\begin{lemma} \label{lem:recurrence}
	\begin{equation*}
    d(i, s, k) = 
    \min(\{d(i-1,s,k)\} \cup \{j_t \mid 0 \leq t < r \})
  \end{equation*}
  holds,
	where $j_t$ is the smallest position $j$ in $B[d(i-1,s-1,t)+1..n]$ such that
	$A[i] = B[j]$, and there exists a string $Z$
	which satisfies $\prop(i-1,s-1,t)$ and $\overlap{ZA[i]} = k$
	(if no such $Z$ exists for $t$, then $j_t = n+1$).
\end{lemma}

\begin{proof}
	We show the following inequations to prove this lemma;
	\begin{enumerate}
		\item $d(i, s, k) \leq \min(\{d(i-1,s,k)\} \cup \{j_t \mid 0 \leq t < r \})$,
		\item $d(i, s, k) \geq \min(\{d(i-1,s,k)\} \cup \{j_t \mid 0 \leq t < r \})$.
	\end{enumerate}

	We start from proving the first inequation.
	By the definition of $d$, $d(i,s,k) \leq d(i-1,s,k)$ always holds.
	If $\{j_t \mid 0 \leq t < r \} = \emptyset$, then the first inequation holds.
	We assume that $\{j_t \mid 0 \leq t < r \} \neq \emptyset$, 
	and $j_{t_1}$ is in the set ($0 \leq t_1 < r$).
	Then, there exists a subsequence $Z_1$ of $B[1..d(i-1,s-1,t_1)]$ which satisfies $\prop(i-1,s-1,t_1)$.
	Since $A[i] = B[j_{t_1}]$ and $j_{t_1} > d(i-1,s-1,t_1)$, 
	$Z_1A[i]$ is a subsequence of $B[1..j_{t_1}]$ that satisfies $\prop(i,s,k)$ and $\overlap{Z_1A[i]}=k$.
	This implies that $d(i,s,k) \leq j_{t_1}$.
	Thus, the first inequation holds.

	Suppose that the second inequation does not hold, namely,
	\begin{equation}
		\label{assumption}
		d(i, s, k) < \min(\{d(i-1,s,k)\} \cup \{j_t \mid 0 \leq t < r \})
	\end{equation}
	holds.
	If $d(i,s,k) = n+1$, then the above inequation does not hold.
	Now we consider the case $d(i,s,k) < n+1$.
	By the definition of $d$, 
	there exists a subsequence $Z_2$ of $B[1..d(i,s,k)]$ that satisfies $\prop(i,s,k)$.
	Let $Z_2' = Z_2[1..|Z_2|-1]$.
	Then, $Z_2'$ is a length $s-1$ subsequence of $A[1..i-1]$ which does not have $P$ as a substring.
	Since $Z_2'$ satisfies $\prop(i-1,s-1,\overlap{Z_2'})$,
	$d(i-1,s-1,\overlap{Z_2'}) < d(i,s,k)$ holds.
	Moreover, 
	$\overlap{Z_2'B[d(i,s,k)]} = k$ holds.
	If $A[i] = B[d(i,s,k)]$,
	then, $j_{\overlap{Z_2'}} \leq d(i,s,k)$ holds.
	This fact contradicts Inequation~(\ref{assumption}).
	Now we can assume that $A[i] \neq B[d(i,s,k)]$.
	This implies that $Z_2$ is a common subsequence of $A[1..i]$ and $B[1..d(i,s,k)-1]$, 
	or a common subsequence of $A[1..i-1]$ and $B[1..d(i,s,k)]$.
	The first case implies a contradiction by the definition of $d$.
	The second case implies that $d(i,s,k) = d(i-1,s,k)$, a contradiction.
	Thus, $d(i, s, k) \geq \min(\{d(i-1,s,k)\} \cup \{j_t \mid 0 \leq t < r \})$ holds.
\end{proof}

\section{Algorithm}\label{sec:algorithm}

In this section, we show how to compute \streclcs~by using Lemma~\ref{lem:recurrence}.
We mainly explain our algorithm to compute the length of an STR-EC-LCS
(we will explain how to compute an STR-EC-LCS at the end of this section).

To use Lemma~\ref{lem:recurrence}, 
we need $d(i-1,s,k)$ and $d(i-1,s-1,t)$ for all $0 \leq t < r$ for computing $d(i,s,k)$.
We compute our table for every diagonal line from upper left to lower right 
in left-to-right order.
In each step of our algorithm, we will fix $0 \leq i, s \leq m$
(we use $(i,s)$ to denote the step for fixed $i$ and $s$).
Then we compute $d(i,s,k)$ for any $0 \leq k < r$ in the step. 
We can see from a simple observation that $d(i,s,k) = n+1$ holds 
for any input strings if $i < s$ (since no STR-EC-LCS of length $s$ exists).
Thus, we do not compute $d(i,s,k)$ explicitly such that $i < s$.
We also describe this strategy in Fig.~\ref{fig:comp_order}.

\begin{figure}[t]
  \centerline{\includegraphics[width=0.7\linewidth]{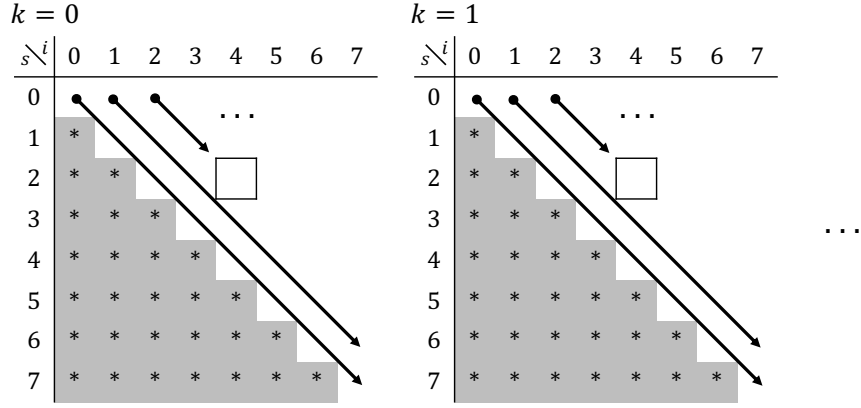}}
  \caption{
  	This figure shows the order of computation for table $d$.
  	For each table (i.e., for each $k$), 
		we do not need to compute the 
		lower left part (satisfying $s > i$).
  	We start from computing values on the leftmost arrow for each table.
  	In each step $(i,s)$, we compute $d(i,s,k)$ for all tables
  	(for instance, squared values in the figure will be computed in the same step).
  }\label{fig:comp_order}
\end{figure}

Now we consider how to compute $d(i,s,k)$ for any $0 \leq k < r$.
Let $Z(i,s,k)$ be a subsequence of $B[1..d(i-1,s-1,k)]$ satisfying $\prop(i-1,s-1,k)$.
Due to Lemma~\ref{lem:recurrence}, 
string $Z(i,s,k)A[i]$ is a witness for value $d(i,s,\overlap{Z(i,s,k)A[i]})$ 
if a (leftmost) position $j$ in $B[d(i-1,s-1,k)+1..n]$ such that $A[i] = B[j]$ exists.
For any $i,s,k$, 
let $J(i,s,k)$ denote the position $j$ described above.
Thus, we can compute $d(i,s,k)$ for any $k$ in step $(i,s)$ as follows.
\begin{enumerate}
	\item Set $d(i-1,s,k)$ as the initial value for $d(i,s,k)$ for each $k$.
	\item Compute $J(i,s,k)$ and $\overlap{Z(i,s,k)A[i]}$ for each $k$.
	\item If $J(i,s,k) < d(i,s,\overlap{Z(i,s,k)A[i]})$, then update $d(i,s,\overlap{Z(i,s,k)A[i]})$ to $J(i,s,k)$.
\end{enumerate}

Lemma~\ref{lem:recurrence} and the above discussion ensure the correctness of this algorithm.
Next we show how to do these operations efficiently.
We use the following two data structures.

\begin{definition}
	For any position $j$ in $B$ (i.e., $j \in [1,n]$) and any character $\alpha \in \Sigma$,
	\begin{equation*}
		\nextchar(j,\alpha) = \min\{q \mid B[q] = \alpha, q \geq j \}.
	\end{equation*}
\end{definition}

\begin{definition}
	For any position $t$ in $P$ (i.e., $t \in [0,r-1]$) and any character $\alpha \in \Sigma$,
	\begin{equation*}
		\newoverlap(t,\alpha) = \overlap{P[1..t]\alpha}.
	\end{equation*}
\end{definition}

At the second operation, we need to compute $J(i,s,k)$.
$J(i,s,k)$ is the index of the leftmost occurrence of $A[i]$ in $B[d(i-1,s-1,k)+1..n]$.
We can compute the occurrence by using $\nextchar$, 
namely, $J(i,s,k) = \nextchar(d(i-1,s-1,k)+1,A[i])$.

Moreover, we need to compute $\overlap{Z(i,s,k)A[i]}$.
We know that $\overlap{Z(i,s,k)} = k$, namely, $Z(i,s,k)$ has $P[1..k]$ as a suffix.
By the definition of $\overlap{\cdot}$, 
$\overlap{S}+1 \geq \overlap{S\alpha}$ holds for any string $S$ and $\alpha \in \Sigma$.
This implies that $\overlap{Z(i,s,k)A[i]} = \overlap{P[1..t]A[i]}$.
Thus, we can compute $\overlap{Z(i,s,k)A[i]}$ by using $\newoverlap(\cdot)$, 
namely, $\overlap{Z(i,s,k)A[i]} = \overlap{P[1..t]A[i]} = \newoverlap(t,A[i])$.

We can easily compute $\nextchar$ in linear time and space 
(we give a pseudo-code in Algorithm~\ref{alg:nextchar}).
$\newoverlap$ was introduced in~\cite{STRECLCS_Wang_2013} (as table $\lambda$).
They also showed that this table can be computed in linear time and space 
(we give a pseudo-code in Algorithm~\ref{alg:newoverlap}).

\begin{algorithm2e}[th]
  \caption{Construction for $\nextchar$}
  \label{alg:nextchar}

	\KwIn{String $B$ of length $n$, Alphabet $\Sigma$}
	\KwOut{$\nextchar$}  

	\lForEach{character $\alpha \in \Sigma$}{$\nextchar(n,\alpha)=n+1$}\;
	\For{$j=n-1$ \KwTo $0$}{
	 	\ForEach{$\alpha \in \Sigma$}{
			\lIf{$\alpha=B[j+1]$}{$\nextchar(j,\alpha)=j+1$}\;
			\lElse{$\nextchar(j,\alpha)=\nextchar(j+1,\alpha)$}\;
	 	}
	}
  
  \KwRet{$\nextchar$}
\end{algorithm2e}

\begin{algorithm2e}[th]
  \caption{Construction for $\newoverlap$}
  \label{alg:newoverlap}

  \KwIn{String $P$ of length $r$, Alphabet $\Sigma$}
  \KwOut{$\newoverlap$}

  $\kmp(0) \leftarrow -1$\;
  $\kmp(1) \leftarrow 0$\;
  $k \leftarrow 0$\;
  \For{$i=2$ \KwTo $r$}{
		\lWhile{$k \geq 0$ and $P[k+1] \neq P[i]$}{
			$k \leftarrow \kmp(k)$\;
		}
		$k \leftarrow k+1$\;
		$\kmp(i) \leftarrow k$\;
  }
	$\newoverlap(0,P[1]) \leftarrow 1$\;
  \ForEach{$\alpha \in \Sigma - \{P[1]\}$}{
		$\newoverlap(0,\alpha) \leftarrow 0$\;
  }
  \For{$k=1$ \KwTo $r-1$}{
    \lForEach{$\alpha \in \Sigma$}{
      \lIf{$\alpha = P[k+1]$}{
      	$\newoverlap(k,\alpha) \leftarrow k+1$\;
      }
      \lElse{
				$\newoverlap(k,\alpha) \leftarrow \newoverlap({\rm kmp}(k),\alpha)$\;      	
      }
    }
  }

  \KwRet ${\rm next}_\sigma$
\end{algorithm2e}

We have finished describing how to compute $d$.
This algorithm computes $O(m^2r)$ values (i.e., the size of the table $d$).
We can see that every operation can be done in constant time.
Thus, this algorithm takes $O(n|\Sigma|+m^2r)$ time and space.
This complexity is similar to Wang et al.s' result (algorithm described in Section~\ref{subsec:Wang}).
We can modify our algorithm to compute $d$ more efficiently by using the following two observations.

\begin{observation}\label{ob:right-upper}
	Assume that we have already computed table $d$ until the $i$-th diagonal line 
	(i.e., the diagonal line which has $d(i,0,\cdot)$).
	Let $s'$ be the lowest row which has a value smaller than $n+1$.
	Then, we do not need to compute the last $s'+1$ diagonal lines 
	since these diagonal lines do not make better candidates for \streclcs.
\end{observation}

\begin{observation}\label{ob:stop}
	If $d(i,s,k) = n+1$ for all $k$, 
	then $d(i+1,s+1,k) = \ldots = d(i+(m-i),s+(m-i),k) = n+1$ holds for any $k$.
\end{observation}

Thanks to the above observations, 
the number of values which we need to compute is $O((L+1)(m-L+1)r)$ 
where $L$ is the length of STR-EC-LCS (see also Fig.~\ref{fig:complexity}).

\begin{figure}[t]
    \centerline{\includegraphics[width=0.4\linewidth]{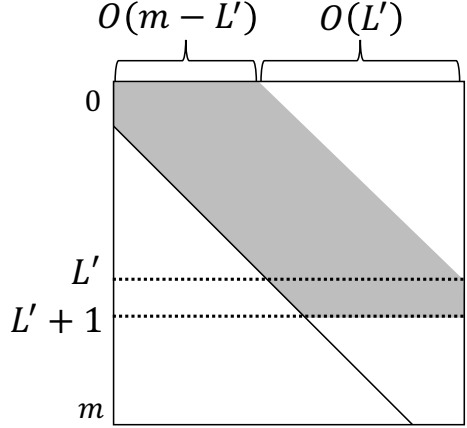}}
    \caption{
    	This is a table for some $k$.
    	Due to Observations~\ref{ob:right-upper} and~\ref{ob:stop}, 
    	we do not need to compute values in white part 
    	(there might exist positions which do not need their values).
    	The maximum number of values which we need to compute 
    	(namely, the total area of the $r$ gray parts) is $O((L+1)(m-L+1)r)$.
    }\label{fig:complexity}
\end{figure}

Finally, we discuss how to store $d$.
We consider computing the $i$-th diagonal line (i.e., $d(i,0,k), \ldots, d(i+(m-i),m-i,k)$).
Suppose that $d(i,0,k), \ldots, d(i+t-1,t-1,k)$ have already been computed.
Then, we store these values by using an array of size $2^{\lceil \log t \rceil}$.
If the array filled with values for the line 
(i.e., $d(i+2^{\lceil \log t \rceil}-1,2^{\lceil \log t \rceil}-1,k) < n+1$ for some $k$),
we make new array of size $2^{\lceil \log t \rceil + 1}$ 
for values $d(i,0,k), \ldots, d(i+2^{\lceil \log t \rceil + 1}-1,2^{\lceil \log t \rceil + 1}-1,k)$
on the line.
By Observation~\ref{ob:stop}, we will compute at most $L+2$ values for each line, 
the total length of arrays for each line is $O(L)$, where $L$ is the length of an STR-EC-LCS.
Therefore, we can compute the length of an STR-EC-LCS in $O(n|\Sigma|+(L+1)(m-L+1)r)$ time and space.

\noindent \textbf{Computing an STR-EC-LCS.}
If we want to compute an STR-EC-LCS, we store a pair $(s',k')$ for every $d(i,s,k)$.
The pair $(s',k')$ represents that $d(i,s,k)$ was given by $d(i-1,s',k')$.
By using these information, we can compute an STR-EC-LCS from right to left.
We show an example in Fig.~\ref{fig:backtrack}.

\begin{figure}[t]
    \centerline{\includegraphics[width=0.8\linewidth]{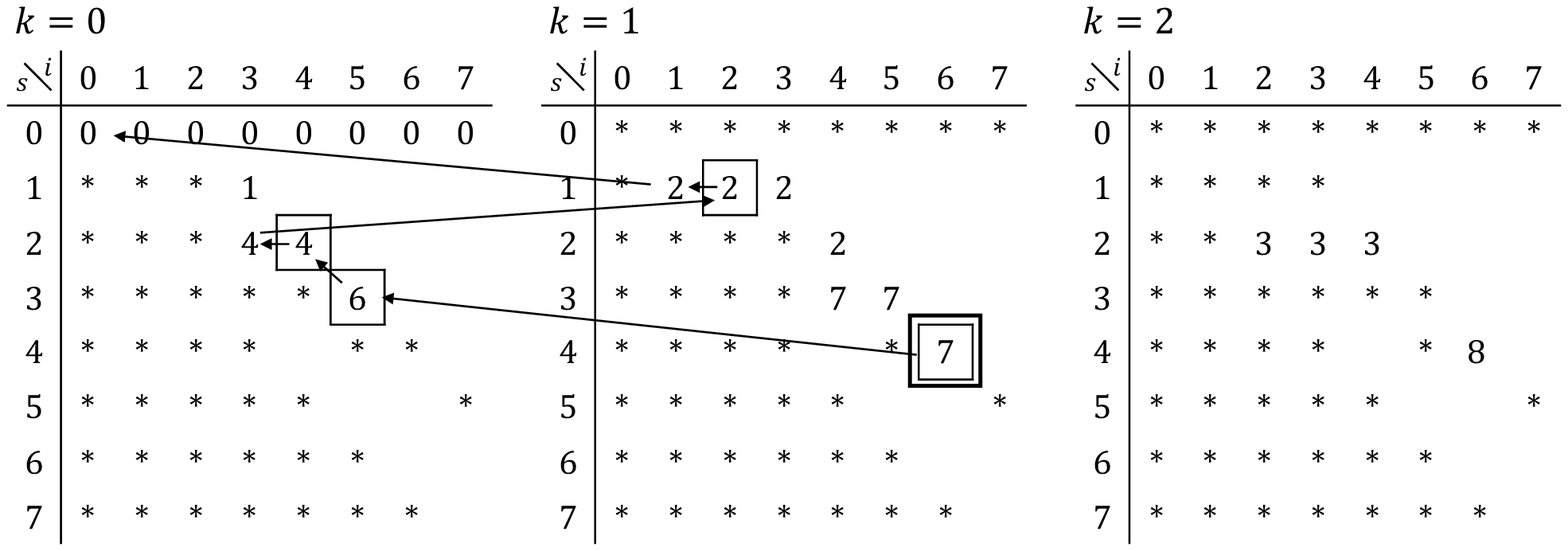}}
    \caption{
    	In this figure, an arrow represents additional information for backtracking.
    	For instance, $d(6,4,1) = 7$ was given by $d(5,3,0) = 6$ while computing $d$.
    	We can get an STR-EC-LCS $\mathtt{abca}$ of $A = \mathtt{aabacab}, B = \mathtt{baabbcaa}$, 
    	and $P = \mathtt{aab}$.
    }\label{fig:backtrack}
\end{figure}

Since we can store $(s',k')$ in constant time and space for each $d(i,s,k)$, 
and compute an STR-EC-LCS in $O(m)$ time, 
we can get the following main result.

\begin{theorem}
	For given strings $A, B$ and $P$,
	we can compute an STR-EC-LCS in $O(n|\Sigma| + (L+1)(m-L+1)r)$ time and space
  where $m, n, r$ and $L$ are the length of $A, B, P$ and the STR-EC-LCS, respectively.
\end{theorem}

\section*{Acknowledgments}
This work was supported by JSPS KAKENHI Grant Numbers JP18K18002 (YN), JP17H01697 (SI), JP16H02783 (HB), JP18H04098 (MT), and by JST PRESTO Grant Number JPMJPR1922 (SI).

\bibliographystyle{abbrv}
\bibliography{ref}

\end{document}